\documentclass[11pt]{article}
\usepackage{CJK}
\usepackage{amsfonts}
\usepackage{mathrsfs}
\usepackage{bbm}
\usepackage{amsfonts}
\usepackage{amssymb,amsmath,graphicx}
\usepackage{amsthm}
\usepackage{float}
\usepackage[colorlinks=true]{hyperref}
\hypersetup{urlcolor=blue, citecolor=red}

\newtheorem{theorem}{Theorem}[section]
\newtheorem{Legendre theorem}{Legendre Theorem}[section]
\newtheorem{lemma}[theorem]{Lemma}
\newtheorem{corollary}[theorem]{Corollary}
\newtheorem{example}[theorem]{Example}
\newtheorem{remark}[theorem]{Remark}

\newtheorem{proposition}[theorem]{Proposition}
\newtheorem{definition}[theorem]{Definition}

\allowdisplaybreaks

\begin{document}

	\title{New Constructions of Full Flag Codes Based on Partial Spreads}
	\author{Xiang Han \and Xinran Li \and Gang Wang\textsuperscript{$^*$}  }
\date{\small College of Science, Civil Aviation University of China, 300300, Tianjin, China. \\ E-mail:xhang@cauc.edu.cn; lixr@163.com; gwang06080923@mail.nankai.edu.cn.\\
$^*$Corresponding author}

	\maketitle {\bf Abstract.} Flag codes are a class of multishot network codes comprising sequences of nested subspaces (flags) within the vector space $\mathbb{F}_q^n$, where $q$ is a prime power. In this paper, we propose a family of constructions for full flag codes based on partial spreads. The distances of this family include maximum distance (optimum distance flag codes), second-maximum distance (quasi-optimum distance flag codes), as well as other feasible  values. The structure of these flag codes resembles that of a \textquotedblleft sandwich", consisting of one layer of companion matrix and two layers of partial spreads. Furthermore, we present an efficient decoding algorithm for these codes.



	\renewcommand{\theequation}{\thesection.\arabic{equation}}
	\catcode`@=11 \@addtoreset{equation}{section} \catcode`@=12
	\maketitle{}
	
	%%%%%%%%%%%%%%%%%%%%%%%%%%%%%%%%%%%%%%%%%%%%%%%%%%%%%%%%%%%%%%%%%%%%%%%%%%%%%%%%%%%%%%%%%%%%%%%%%%
	
\section{Introduction}
	
	Random network coding was initially introduced in \cite{AC} and serves as a groundbreaking transmission strategy for information networks. The network model is represented as an acyclic directed graph that may include multiple senders and receivers. Each intermediate node is capable of performing a random linear combination of the input vectors and transmitting the resulting output, thereby enhancing the information rate of the network. \textit{Subspace codes} were first introduced by Koetter and Kschischang in \cite{KK} as an algebraic framework for random network coding. This approach employs vector spaces rather than individual vectors as codewords. A subspace code is formally defined as a collection of subspaces within a vector space $\mathbb{F}_q^n$ over a finite field $\mathbb{F}_q$, where $q$ is a prime power.
	
	In recent decades, researchers have focused significant attention on a specific type of subspace codes known as \textit{constant dimension codes}, which are distinguished by the property that all codewords (subspaces) share the same dimension. In \cite{GR} and \cite{MGR}, two special classes of constant dimension codes, referred to as \textit{spread codes} and \textit{partial spread codes}, were introduced.
	
	The concept of constant dimension codes was subsequently extended to \textit{multishot codes}, i.e., codes where the codeword is a sequence of subspaces (see \cite{NU}). On the other hand, \textit{flag codes} were initially introduced in \cite{LNV} as multishot codes with a rich algebraic structure. These refer to sequences of nested subspaces with specified dimensionalities (type vectors of flags) in $\mathbb{F}_q^n$. Furthermore, \cite{ANS1} proposed the notion of \textit{projected codes} for flag codes, i.e., constant dimension codes derived by projecting flags from a flag code. In recent years, numerous studies have investigated how the various structural properties of a flag code influence its projected codes and the corresponding feedback effects (see \cite{AN4} and \cite{ANS3}).
	
	The properties of \textit{optimal distance flag code} (ODFC) were thoroughly examined in \cite{ANS1}, where the first method for constructing such codes utilizing spread codes was introduced, along with a two-step decoding algorithm. In \cite{K}, lower and upper bounds on the maximum possible number of codewords for various flag codes with given parameters were established, and a construction for optimal distance full flag codes with $n=2k+1$ was provided. Additional studies related to optimal distance flag codes can be found in \cite{CY,ANS2,AN1}.
	
	More recently, \textit{quasi-optimum distance cyclic orbit flag codes}, i.e., codes that achieve a \textquotedblleft second best\textquotedblright\ value of the flag distance in the orbital case, were studied in \cite{AN2}. Subsequently, in \cite{AN3}, the authors extended their investigation to \textit{quasi-optimum distance flag code}(QODFC) within a general framework, with a focus on the interplay between the flag code and its projected codes. They characterized QODFC in terms of a specific projected flag code, provided a systematic construction method for it under arbitrary type vectors, and generalized the construction approach to accommodate arbitrary distances within a controlled range. They also posed the problem of constructing a QODFC with the maximum cardinality.
	
	 The primary objective of this paper is not only to propose a construction for QODFC with maximum cardinality, but also to further extend this construction and develop a corresponding decoding algorithm. To solve these problems, we develop a family of constructions for full flag codes based on partial spreads, motivated by the construction of partial spread in \cite{GR} and the ODFC based on spread in \cite{ANS1}. These constructions are called \textquotedblleft sandwich", including one layer of companion matrix $M$ and two layers of partial spreads related to $M$. According to the difference of \textquotedblleft sandwich" constructions, we can derive ODFC, QODFC and the flag codes with other distance values. Specially, the cardinality of QODFC can reach the maximum value under some condition. In fact, the cardinalities of the QODFC and other full flag codes generated by the \textquotedblleft sandwich" structure, with varying distance values, are greater than that of the full flag codes exhibiting the same distance in \cite{AN3}. Furthermore, building upon the two-step decoding algorithm presented in \cite{ANS1}, we have extended it to a three-step decoding algorithm and applied this enhancement to the decoding process of the full flag codes produced by the \textquotedblleft sandwich" structure.
	
	The structure of this paper is organized as follows. Section 2 provides fundamental background on constant dimension codes and flag codes. Section 3 explores the relationship between the cardinality and distance of full flag codes and their projected codes, deriving an upper bound on the cardinality of a specific class of full flag codes. Section 4 focuses on proposing a family of constructions for full flag codes based on partial spreads. Specifically, Definition \ref{D4.7} introduces a structural framework for flag codes that resembles a \textquotedblleft sandwich", consisting of three layers of matrices. These flag codes can then be categorized into three distinct classes according to their respective distances: ODFC, QODFC and flag codes with other possible distance values. Finally, Section 5 presents an efficient decoding algorithm tailored for these codes on erasure channel.

\section{Preliminaries}
\subsection{Constant dimension codes}
In this section, we present the essential concepts pertaining to constant dimension codes. The definitions below are all derived from \cite{KK}.
Let $q$ be a prime power and $\mathbb{F}_q$ the finite field with $q$ elements. Let $\mathbb{F}_q^n$ denote the $n$-dimensional vector space over $\mathbb{F}_q$. For a given dimension $1 \leqslant k < n$, the set of all $k$-dimensional vector subspaces of $\mathbb{F}_q^n$ is denoted by $\mathcal{G}_q(k,n)$, referred to as the \textit{Grassmannian}. The subspace distance is defined as
 \begin{align}\label{F2.1}
 d_S(\mathcal{U},\mathcal{V}):=&\dim(\mathcal{U}+\mathcal{V})-\dim(\mathcal{U}\cap\mathcal{V})\nonumber\\
 =&2(k-\dim(\mathcal{U}\cap\mathcal{V}))=2(\dim(\mathcal{U}+\mathcal{V})-k),	
 \end{align}
 where $\mathcal{U}$, $\mathcal{V}$ $\in$ $\mathcal{G}_q(k,n)$.
 A \textit{constant dimension code} of dimension $k$ is a nonempty subset $\mathcal{C} \subseteq \mathcal{G}_q(k,n)$, and the minimum distance of a constant dimension code $\mathcal{C}$ is defined as
 \begin{align}
 d_S(\mathcal{C}):=\min\{d_S(\mathcal{U},\mathcal{V})\ |\ \mathcal{U}, \mathcal{V} \in \mathcal{C},\ \mathcal{U}\neq\mathcal{V}\}.	
 \end{align}
 In the case where $|\mathcal{C}| = 1$, we define $d_S(\mathcal{C}) = 0$. Thus, the minimum distance of a constant dimension code $\mathcal{C}$ is
 \begin{align}
 d_S(\mathcal{C})\leqslant\min\{2k,2(n-k)\}.	
 \end{align}

 A \textit{partial $k$-spread} of $\mathbb{F}_q^n$ is a collection of $k$-dimensional vector subspaces of $\mathbb{F}_q^n$ that intersect trivially. If the union of the elements in the partial $k$-spread is $\mathbb{F}_q^n$, then we refer to it as a \textit{$k$-spread}.

 Let $A_q(n,d;k)$ denote the maximum size of a constant dimension code with dimension $k$ and minimum distance $d$ in $\mathbb{F}_q^n$. The following results provide some information about $A_q(n,d;k)$, which are utilized in this paper.

\begin{lemma}\rm\cite [Lemma 7]{GR}\label{L2.1}
Let $\mathcal{C} \subseteq \mathcal{G}_q(k,n)$ be a partial spread code. Then the minimum distance of $\mathcal{C}$ is $2k$, and its maximum size is
\begin{align}
A_q(n,2k;k)\leqslant \lfloor\frac{q^n-1}{q^k-1}\rfloor.
\end{align}
\end{lemma}

\begin{lemma}\rm\cite [Theorem 5]{NS}\label{L2.2}
Let $r$ denote the remainder obtained when dividing $n$ by $k$. If $k>\frac{q^r - 1}{q - 1}$, then
\begin{align}
A_q(n,2k;k)=\frac{q^n-q^{k+r}}{q^k-1}+1.
\end{align}
\end{lemma}

\subsection{Flag codes}
In this section, we introduce the fundamental concepts related to flag codes. The definitions below are all derived from \cite{AN3}.

\begin{definition}
Given integers $0 < t_1 < \cdots < t_r < n$, a flag of type $\mathbf{t} = (t_1, \dots, t_r)$ on $\mathbb{F}_q^n$ is a sequence $\mathcal{F} = (\mathcal{F}_1, \dots, \mathcal{F}_r)$ of $\mathbb{F}_q$-subspaces of $\mathbb{F}_q^n$ satisfying the following conditions:
\begin{itemize}
\item[(1)] $\mathcal{F}_1 \subsetneq \cdots \subsetneq \mathcal{F}_r \subsetneq \mathbb{F}_q^n$, 	
\item[(2)] $\dim(\mathcal{F}_i) = t_i$ for all $1 \leqslant i \leqslant r$.	
\end{itemize}
\end{definition}

\begin{definition}
The set of all flags of a given type vector $\mathbf{t} = (t_1, \dots, t_r)$ is referred to as the flag variety of type $\mathbf{t}$, denoted by $\mathcal{F}_q(\mathbf{t},n)$. A flag code $\mathcal{C}$ of type $\mathbf{t} = (t_1, \dots, t_r)$ on $\mathbb{F}_q^n$ is a nonempty subset of the flag variety $\mathcal{F}_q(\mathbf{t},n)$. Furthermore, a flag code $\mathcal{C}$ of type $\mathbf{t} = (1, \dots, n-1)$ is called a full flag code.
\end{definition}

The flag distance between two flags $\mathcal{F}=(\mathcal{F}_1,\dots,\mathcal{F}_r)$ and  $\mathcal{F}'=(\mathcal{F}'_1,\dots,\mathcal{F}'_r)$ of type $\textbf{t} = (t_1,\dots,t_r)$ on $\mathbb{F}_q^n$ is
\begin{align}
	d_f(\mathcal{F}, \mathcal{F}'):=\sum_{i=1}^r d_S(\mathcal{F}_i, \mathcal{F}'_i),
\end{align}
and the minimum distance of a flag code $\mathcal{C}$ of type $\textbf{t} = (t_1,\dots,t_r)$ on $\mathbb{F}_q^n$ is
\begin{align}\label{F2.7}
	d_f(\mathcal{C}):=\min \{ d_f(\mathcal{F}, \mathcal{F}')\mid\mathcal{F}, \mathcal{F}' \in \mathcal{C}, \mathcal{F}\neq \mathcal{F}' \},
\end{align}
whenever $|\mathcal{C}| \geqslant 2$. In case $|\mathcal{C}| = 1$, we put $d_f(\mathcal{C}) = 0$.

The concepts of optimum distance flag codes and quasi-optimum distance flag codes represent significant areas of study within the field of flag codes, which are defined as follows.

\begin{lemma}\rm\cite{AN3}\label{L2.5}
The value $d_f(\mathcal{C})$ is an even integer satisfying $0 \leqslant d_f(\mathcal{C}) \leqslant D^{(\mathbf{t},n)}$, where
\begin{align}
	D^{(\mathbf{t},n)}:=2\left( \sum_{t_i\leqslant\frac{n}{2}} t_i+\sum_{t_i>\frac{n}{2}} (n-t_i)\right) .
\end{align}
If $d_f(\mathcal{C}) = D^{(\mathbf{t},n)}$, then $\mathcal{C}$ is called an optimum distance flag code (ODFC).

Moreover, if $\mathcal{C}$ is a full flag code of type $\textbf{t} = (1, \dots, n-1)$ on $\mathbb{F}_q^n$, then $D^{(\mathbf{t},n)}$ can be simplified to $D^{(n)}$ and
\begin{align}\label{F2.9}
	D^{(n)}=\left\{ \begin{matrix}
		\frac{n^2-1}{2},   &  n\ \text{is}\  \text{odd},\\
		\frac{n^2}{2},   &  n\ \text{is}\   \text{even}. \\
    \end{matrix}	
    \right.
\end{align}
\end{lemma}

\begin{definition}\label{D2.3}
Given the type vector $\mathbf{t} = (t_1,\dots,t_r)$, a flag code $\mathcal{C}\subseteq \mathcal{F}_q(\mathbf{t},n)$ is said to be a quasi-optimum distance flag code (QODFC) if $d_f(\mathcal{C})=D^{(\mathbf{t},n)}-2$.
\end{definition}

The flag code is closely related to its projected codes, the definition of which is as follows.

\begin{definition}\label{D2.7}
Let $\mathcal{C}$ be a flag code of type $\mathbf{t} = (t_1,\dots,t_r)$ on $\mathbb{F}_q^n$. For every $1 \leqslant i \leqslant r$, the $i$-th projected (subspace) code of $\mathcal{C}$ is the constant dimension code
\begin{align}
\mathcal{C}_i:=\{\mathcal{F}_i\mid \mathcal{F}=(\mathcal{F}_1,\dots,\mathcal{F}_i,\dots,\mathcal{F}_r)\in \mathcal{C} \} \subseteq \mathcal{G}_q(t_i,n).	
\end{align}	
\end{definition}

Let $\mathcal{C}$ be a flag code of type $\mathbf{t} = (t_1,\dots,t_r)$ on $\mathbb{F}_q^n$. Whenever they appear in \textbf{t}, denote the special dimensions
\begin{align}
	t_L := \max \{t_i\mid2t_i \leqslant n\} \ \ \mathrm{and} \ \ t_R :=\min\{t_i\mid2t_i\geqslant n\}.
\end{align}
If $\frac{n}{2}$ is a dimension in the type vector $\mathbf{t}$, then $t_L = t_R = \frac{n}{2}$. Moreover, if every dimension is upper (resp. lower) bounded by $\frac{n}{2}$, then $L = r$ and $R$ is not defined (resp. $R = 1$ and $L$ is not defined). In all other cases, these dimensions $t_L$ and $t_R$ exist, and they are distinct and consecutive.

The methods for determining ODFC and QODFC are given by the following conclusions.

\begin{lemma}\rm\cite [Theorem 4.8]{NS1}\label{L2.4}
Let $\mathcal{C}$ be a flag code of type $\textbf{t} = (t_1,\dots,t_r)$ on $\mathbb{F}_q^n$ and consider the indices $L$ and $R$. The following statements are equivalent:
\begin{itemize}
\item[(1)] The flag code $\mathcal{C}$ is an ODFC,
\item[(2)] $\mathcal{C}_L$ and $\mathcal{C}_R$ are constant dimension codes of maximum distance with cardinality $|\mathcal{C}_L|=|\mathcal{C}_R|=|\mathcal{C}|$.
\end{itemize}
\end{lemma}

\begin{lemma}\rm\cite [Corollary 3.6]{AN3}\label{L2.9}
Let $\mathcal{C}$ be a flag code of type $\textbf{t} = (t_1,\dots,t_r)$ on $\mathbb{F}_q^n$, and suppose that $|\mathcal{C}| = |\mathcal{C}_1| = \cdots = |\mathcal{C}_r|$. Then $\mathcal{C}$ is a QODFC if and only if, the following conditions hold:
\begin{itemize}
\item[(1)] $\mathcal{C}_{L-1}$ and $\mathcal{C}_{R+1}$ are subspace codes of maximum distance,
\item[(2)] $\mathcal{C}_{(L,R)}$ is a QODFC.	
\end{itemize}
\end{lemma}

\section{Characterizations for full flag codes}
This section will systematically examine the relationship between the cardinality and distance of full flag codes, as well as their projected codes. Additionally, it will derive an upper bound on the cardinality for a specific class of full flag codes.

As shown in Lemma \ref{L2.4} and \ref{L2.9}, if the flag code $\mathcal{C}$ is an ODFC or a QODFC, the projected codes $\mathcal{C}_i$ have a significant impact on both the distance and cardinality of $\mathcal{C}$. This influence persists for other flag codes as well. The following lemma, as a generalization of \cite[Theorem 5.7]{AN3}, explicitly shows this influence.

\begin{lemma}\label{L3.1}
Let $\mathcal{C}$ be a full flag code of type $\textbf{t} = (1,\dots,n-1)$ on $\mathbb{F}_q^n$ with distance $d_f(\mathcal{C})=D^{(n)}-2l$, where $0\leqslant l< \frac{n-1}{2}$ is an integer. Then 
\begin{itemize}
\item[(1)] the distance of the projected code $\mathcal{C}_i$ is maximized, i.e., $d_S(\mathcal{C}_i)=\min\{2i,2(n-i)\}$ for all $i\in \{1,2,\dots,L-l,R+l,\dots,n-1\}$, and
\end{itemize}
\begin{itemize}
\item[(2)] $|\mathcal{C}_i|=|\mathcal{C}|$ for all $i\in \{1,\dots,n-1\}$.	
\end{itemize}

\end{lemma}

\begin{proof}
(1) Suppose the assertion is false, that $d_S(\mathcal{C}_i)$ is not at its maximum for some $i\in \{1,\dots,L-l,R+l,\dots,n-1\}$. Then we could find a pair of distinct full flags $\mathcal{F},\mathcal{F}'\in\mathcal{C}$, such that $d_S(\mathcal{F}_i,\mathcal{F}'_i)<\min\{2i,2(n-i)\}$. The proof will be divided into two steps.

Case 1: We first consider $i\leqslant L-l$. From (\ref{F2.1}) it shows that $\mathcal{F}_i\cap\mathcal{F}'_i\neq\{0\}$, and as $\mathcal{F}_i\subset\mathcal{F}_{i+1}$ and $\mathcal{F}'_i\subset\mathcal{F}'_{i+1}$, we have $\mathcal{F}_{i+1}\cap\mathcal{F}'_{i+1}\neq\{0\}$. This gives
\begin{align*}
	d_S(\mathcal{F}_{i+1},\mathcal{F}'_{i+1})&=2(i+1-\dim(\mathcal{F}_{i+1}\cap\mathcal{F}'_{i+1}))\\
	&\leqslant2(i+1-1)=2i\\
	&<\min\{2(i+1),2(n-(i+1))\}=2(i+1).
\end{align*}
This leads to $d_S(\mathcal{C}_{i+1})\leqslant2(i+1)-2$. In a similar manner, we can see that
\begin{align*}
	d_S(\mathcal{C}_{i+2})\leqslant2(i+2)-2,\dots,d_S(\mathcal{C}_L)\leqslant2L-2.
\end{align*}
According to (\ref{F2.7}), it can be concluded that
\begin{align*}
	d_f(\mathcal{C})\leqslant D^{(n)}-2(L-i+1)\leqslant D^{(n)}-2(l+1),
\end{align*}
which contradicts that $d_f(\mathcal{C})=D^{(n)}-2l$.

Case 2: We now turn to the case $i\geqslant R+l$. (\ref{F2.1}) leads to  $\mathcal{F}_i+\mathcal{F}'_i\neq\mathbb{F}_q^n$, and since $\mathcal{F}_{i-1}\subset\mathcal{F}_i$ and $\mathcal{F}'_{i-1}\subset\mathcal{F}'_i$, we thus get $\mathcal{F}_{i-1}+\mathcal{F}'_{i-1}\neq\mathbb{F}_q^n$. Hence
\begin{align*}
	d_S(\mathcal{F}_{i-1},\mathcal{F}'_{i-1})&=2(\dim(\mathcal{F}_{i-1}+\mathcal{F}'_{i-1})-(i-1))\\
	&\leqslant2(n-1-(i-1))=2(n-i)\\
	&<\min\{2(i-1),2(n-(i-1))\}=2(n-i+1).
\end{align*}
From this, we have $d_S(\mathcal{C}_{i-1})\leqslant2(n-i+1)-2$. As in the proof of Case 1, it follows that
\begin{align*}
	d_f(\mathcal{C})\leqslant D^{(n)}-2(i-R+1)\leqslant D^{(n)}-2(l+1),
\end{align*}
which is impossible.

It is sufficient to prove that $d_S(\mathcal{C}_i)$ is maximized for every $i\in \{1,2,\dots,L-l,R+l,\dots,n-1\}$, which completes the proof.

(2) On the contrary, suppose that $|\mathcal{C}_i|<|\mathcal{C}|$ for some $i\in \{1,\dots,n-1\}$. Then there exist two distinct flags $\mathcal{F}=(\mathcal{F}_1,\cdots,\mathcal{F}_i,\cdots,\mathcal{F}_{n-1}),\mathcal{F}'=(\mathcal{F}'_1,\cdots,\mathcal{F}'_i,\cdots,\mathcal{F}'_{n-1})\in\mathcal{C}_i$ such that $\mathcal{F}_i=\mathcal{F}'_i$.

If $i\leqslant L$, there would be   
\begin{align*}   
	d_S(\mathcal{F}_i,\mathcal{F}'_i)=0,\  \textrm{and}\  d_S(\mathcal{F}_j,\mathcal{F}'_j)\leqslant 2(j-i)\quad\textrm{for all} \ j\in \{i+1,\dots,L\}.
\end{align*}
Hence
\begin{align*}
	d_f(\mathcal{C})\leqslant d_f(\mathcal{F},\mathcal{F}')\leqslant D^{(n)}-2i(L-i+1).
\end{align*}
Since $0\leqslant l< \frac{n-1}{2}\leqslant L$ and $i\leqslant L$, we have $2i(L-i+1)\geqslant 2L> 2l$, contrary to $d_f(\mathcal{C})=D^{(n)}-2l$.

Similar considerations apply to $i>L$, and this finishes the proof.	
\end{proof}

\begin{remark}
If $n$ is odd and $l=\frac{n-1}{2}=L$, then $|\mathcal{C}_i|$ may be less than $|\mathcal{C}|$ for $i=\{1,n-1\}$. Here we present an example to illustrate it.
\end{remark}

\begin{example}\label{E3.2}
Let $\mathbb{F}_q^n=\mathbb{F}_2^7=\langle e_1,\dots,e_7 \rangle$, where $\{e_1,\dots,e_7\}$ denotes the standard basis of $\mathbb{F}_2^7$. Subsequently, define full flags
\begin{align*}
\mathcal{F}^1=(&\langle e_1\rangle,\langle e_1,e_2\rangle,\langle e_1,e_2,e_3\rangle,\langle e_1,e_2,e_3,e_4\rangle,\langle e_1,\dots,e_5\rangle,\langle e_1,\dots,e_6\rangle),\\
\mathcal{F}^2=(&\langle e_1\rangle,\langle e_1,e_5\rangle,\langle e_1,e_5,e_6\rangle,\langle e_1,e_5,e_6,e_7\rangle,\langle e_1,e_4,e_5,e_6,e_7\rangle,\langle e_1,e_3\dots,e_7\rangle),\\
\mathcal{F}^3=(&\langle e_6+e_7\rangle,\langle e_6+e_7,e_3+e_5\rangle,\langle e_6+e_7,e_3+e_5,e_3+e_4+e_7\rangle,\langle e_1,e_6+e_7,e_3+e_5,e_3+e_4+e_7\rangle,\\
&\langle e_1,e_3,e_5,e_4+e_7,e_6+e_7\rangle,\langle e_1,e_3\dots,e_7\rangle).
\end{align*}

The distances of the flags are given by
\begin{align*}
d_f(\mathcal{F}^1,\mathcal{F}^2)&=0+2+4+6+4+2=18=D^{(7)}-2l,\\
d_f(\mathcal{F}^2,\mathcal{F}^3)&=2+4+6+4+2+0=18=D^{(7)}-2l,\\
d_f(\mathcal{F}^1,\mathcal{F}^3)&=2+4+6+6+4+2=24=D^{(7)},
\end{align*}
where $D^{(7)}=24$ and $l=3$. And the distance of the flag code $\mathcal{C}=\{\mathcal{F}^1,\mathcal{F}^2,\mathcal{F}^3\}$ is
\begin{align*}
d_f(\mathcal{C})=d_f(\mathcal{F}^1,\mathcal{F}^2)=18.
\end{align*}

Obviously, the cardinalities of the projected codes are
\begin{align*}
|\mathcal{C}_1|=|\mathcal{C}_6|=2<|\mathcal{C}|=3.
\end{align*} 	
\end{example}

Next we present an upper bound for the cardinality $|\mathcal{C}|$ of a specific class of full flag codes $\mathcal{C}$. We indicate that $|\mathcal{C}|$ is associated with the upper bound of a particular projected code.

\begin{theorem}\label{T3.5}
Set $n = 2k + r$ with $0 \leqslant r < k$ and $k \geqslant 2$. Let $\mathcal{C}$ denote a full flag code of type $\mathbf{t} = (1, \dots, n-1)$ over $\mathbb{F}_q^n$, with distance $d_f(\mathcal{C}) = \frac{n^2-r^2}{2}$ for $r \in \{0, 1, 2\}$. Then the cardinality of the code $\mathcal{C}$ satisfies $|\mathcal{C}| \leqslant A_q(n, 2k; k)$. Moreover, if $k>\frac{q^r - 1}{q - 1}$, then $|\mathcal{C}| \leqslant q^{k+r}+1$.
\end{theorem}
\begin{proof}
Our proof starts with the discussion on the value of $r$, and is structured into three parts.

Case 1: Consider $r=0$, then $n=2k+r=2k$ is even. Since $D^{(n)}=\frac{n^2}{2}=\frac{n^2-r^2}{2}=d_f(\mathcal{C})$ by (\ref{F2.9}), the full flag code $\mathcal{C}$ constitutes an ODFC.
According to Lemma \ref{L2.4}, the $L$-th projected code $\mathcal{C}_L=\mathcal{C}_k$ achieves the maximum distance $2L$ with cardinality $|\mathcal{C}_L| = |\mathcal{C}|$. We thus get $|\mathcal{C}| = |\mathcal{C}_L| = |\mathcal{C}_k| \leqslant A_q(n, 2k; k)$.
	
Case 2: Consider $r = 1$, then $n = 2k + 1$ is odd, and $D^{(n)} = \frac{n^2 - 1}{2} = d_f(\mathcal{C})$, which implies that $\mathcal{C}$ is also an ODFC with $L = \frac{n-1}{2} = k$ and $R=\frac{n + 1}{2} = k + 1$.
In the same manner, we can see that $|\mathcal{C}| = |\mathcal{C}_L| = |\mathcal{C}_k| \leqslant A_q(n, 2k; k)$.
	
Case 3: We now turn to the case $r=2$. Since $n = 2k + 2$ is odd, it follows that $L=R=k+1$, $D^{(n)} = \frac{n^2}{2}$, and $d_f(\mathcal{C})=\frac{n^2-r^2}{2}=D^{(n)}-2=D^{(n)}-2l$ with $l=1$. By Lemma \ref{L3.1}, the value of $d_S(\mathcal{C}_i)$ is maximized for $i\in \{1,2,\dots,L-l,R+l,\dots,n-1\}$, and $|\mathcal{C}|=|\mathcal{C}_{L-1}|$. Consequently, we deduce that $|\mathcal{C}|=|\mathcal{C}_{L-1}|\leqslant  A_q(n,2(L-1);L-1)=A_q(n,2k;k)$.

According to Lemma \ref{L2.2}, if $k>\frac{q^r - 1}{q - 1}$, then
\begin{align*}
A_q(n,2k;k)=\frac{q^n-q^{k+r}}{q^k-1}+1=\frac{q^{2k+r}-q^{k+r}}{q^k-1}+1=q^{k+r}+1.
\end{align*}
Hence $|\mathcal{C}|\leqslant q^{k+r}+1$.
\end{proof}

In the following section, we will systematically present the constructions of full flag codes $\mathcal{C}$ with distance $d_f(\mathcal{C}) = \frac{n^2-r^2}{2}$ and cardinality $|\mathcal{C}|=q^{k+r}+1$ for every $r \geqslant 0$.

\section{The \textquotedblleft sandwich" constructions}
This section focuses on developing a class of constructions for full flag codes based on partial spreads. This class includes optimum distance flag codes, quasi-optimum distance flag codes, and flag codes with other possible distance values. The structure of these flag codes resembles that of a \textquotedblleft sandwich", consisting of three layers of matrices.

At the beginning of this section, we introduce several notations used throughout this paper. We start by constructing a partial $k$-spread in the space $\mathbb{F}_q^n$.

\begin{lemma}\rm\cite [Ch.\  2.5]{RH1}\label{L4.1}
	\ Let $q$ be a prime power and $\mathbb{F}_q$ the finite field with $q$ elements. Choose a primitive  polynomial $p=x^k+\sum_{i=0}^{k-1}p_ix^i\in\mathbb{F}_q[x]$ of degree $k\geqslant 1$. Define the \textit{companion matrix} of $p$ as
	\begin{align}\label{LA3.1}
		\textbf{M}(p):=\begin{pmatrix}
			0 & 1 & 0 & \cdots & 0\\
			0 & 0 & 1 &   & 0\\
			\vdots&  &   &\ddots& \vdots\\
			0 & 0 & 0 &   & 1\\
			-p_0 & -p_1 & -p_2 &\cdots   & -p_{k-1}\\
		\end{pmatrix}.
	\end{align}
	Then the $\mathbb{F}_q-$ algebra $\mathbb{F}_q[\mathbf{M}(p)]$ is a finite field with $q^{k}$ elements $\{\textbf{O},\textbf{M}(p),\textbf{M}(p)^2,\dots,\textbf{M}(p)^{q^k-1}\}$.
\end{lemma}

In this paper, the primitive polynomial $p$ is fixed and $\textbf{M}(p)$ is briefly denoted by $M$ in the sequel.

The following lemma presents some properties of companion matrix $M$, which will be utilized in the subsequent discussions.

\begin{lemma}\label{L4.2}
Let $M^i$ be the $i-$th power of the companion matrix $M$ for $i = 1, \dots, q^k - 1$.
\begin{itemize}
\item[(1)] Let $v_1, v_2, \dots, v_k$ denote the row vectors of the companion matrix $M$ respectively. Then the row vectors of the matrix $M^i$ can be expressed as $v_1 M^{i-1}, v_2 M^{i-1}, \dots, v_k M^{i-1}$ respectively.
\end{itemize}
\begin{itemize}
\item[(2)] Let $v$ denote the first row vector of $M^i$. Then the row vectors of the matrix $M^i$ are precisely $v, vM, \dots, vM^{k-1}$.
\end{itemize}
\end{lemma}
\begin{proof}
The proof is straightforward according to the definition of companion matrix $M$.	
\end{proof}

\begin{definition}\label{D4.2}
Let $A$ be a $k \times k$ matrix over $\mathbb{F}_q$, and let $1 \leqslant t \leqslant k$ be an integer. Denote by $A^{(t)}$ the first $t$ rows of $A$, and by $A_{(t)}$ the last $t$ rows of $A$. Set the zeroth power of $A$ as $A^{0} := O_k$, where $O_k$ is the $k \times k$ zero matrix.
\end{definition}

\begin{remark}\label{R4.3}
According to Lemma  \ref{L4.1} and Definition \ref{D4.2}, the companion matrix $M$ of a primitive polynomial $p$ of degree $k\geqslant 1$ over $\mathbb{F}_q$ satisfies

\begin{align}
M^i=\left\{ \begin{array}{ll}
	O_k,   &  i=0,\\
	M^i,   &  i=1,\dots, q^k-2,\\
	I_{k}, &  i=q^k-1. \\
\end{array}	
\right.	
\end{align}
\end{remark}

In what follows, we describe the construction of a partial spread on $\mathbb{F}_q^n$, which serves as the foundation for this section. Specifically, it will function as both the first and third layers in the \textquotedblleft sandwich" structure.

From now on, let $n = 2k_1 + r$ with $0 \leqslant r < k_1$ and set $k_2 = k_1 + r$, then $n=k_1+k_2$.

\begin{proposition}\label{P4.4}
Let $M$ be the companion matrix of a fixed primitive polynomial $p$ of degree $k_2 \geqslant 1$. Define the matrices
\begin{equation}
	\begin{aligned}\label{LA3.2.1}
		A[1]&:=\begin{bmatrix}
			O_{k_1} & I_{k_1} & O_{k_1\times r}\\
		\end{bmatrix}_{k_1\times n},\\
		A[i]&:=\begin{bmatrix}
			I_{k_1} & (M^{i-2})^{(k_1)}\\
		\end{bmatrix}_{k_1\times n},\quad i=2,\dots,q^{k_2}+1,
	\end{aligned}
\end{equation}
and the set $\mathbf{A}:= \{A[i]\ |\ i=1,\dots,q^{k_2}+1\}$.

Denote by $\mathcal{A}[i]$ the subspace generated by the matrix $A[i]$, i.e., $\mathcal{A}[i] = \mathrm{rowsp}(A[i])$. Then the set
\begin{align}
	\mathfrak{A} := \{\mathcal{A}[i] \mid i = 1, \dots, q^{k_2} + 1\} \subseteq \mathcal{G}_q( k_1,n)
\end{align}
satisfies the following properties:
\begin{itemize}
\item [(1)] $\mathcal{A}[i]$ is a $k_1$-dimensional subspace of $\mathbb{F}_q^n$ for any $i \in \{1, \dots, q^{k_2} + 1\}$,
\end{itemize}
\begin{itemize}
\item [(2)] $\mathcal{A}[i] \cap \mathcal{A}[j] = \{0\}$ for any $i \neq j \in \{1, \dots, q^{k_2} + 1\}$,
\end{itemize}
\begin{itemize}
\item [(3)] there exist no set $\mathfrak{A}' \subseteq \mathcal{G}_q(k_1,n)$ such that $\mathfrak{A}$ is a proper subset of $\mathfrak{A}'$. Hence, $\mathfrak{A}$ is a partial $k_1$-spread of $\mathbb{F}_q^n$.
\end{itemize}
\end{proposition}

\begin{proof}
(1) Since every matrix $A[i]$ is in row-reduced echelon form and $\mathrm{Rank}(A[i]) = k_1$, the subspace $\mathcal{A}[i]$ is precisely a $k_1$-dimensional subspace of $\mathbb{F}_q^n$ for any $i \in \{1, \dots, q^{k_2} + 1\}$.

(2) Choose two distinct matrices $A[i]$ and $A[j]$ in the set $\mathbf{A}$.

If either $i = 1$ or $j = 1$, without loss of generality, assume $i = 1$. Since
\begin{align*}
	\begin{bmatrix}
		A[i]\\
		A[j]
	\end{bmatrix}=\begin{bmatrix}
		A[1]\\
		A[j]
	\end{bmatrix}=
	\begin{bmatrix}
		O_{k_1} &  I_{k_1} \  \ \ O_{k_1\times r}\\
		I_{k_1} & \ \  (M^{j-2})^{(k_1)} \ \ \\
	\end{bmatrix}_{2k_1\times n},
\end{align*}
it follows that $\mathcal{A}[1] \cap \mathcal{A}[j] = \{0\}$.

If both $i \neq 1$ and $j \neq 1$, then
\begin{align*}
	\begin{bmatrix}
		A[i]\\
		A[j]
	\end{bmatrix}=\begin{bmatrix}
		I_{k_1} & (M^{i-2})^{(k_1)} \\
		I_{k_1} & (M^{j-2})^{(k_1)} \\
	\end{bmatrix}_{2k_1\times n}
	\rightarrow\begin{bmatrix}
		I_{k_1} &  (M^{i-2})^{(k_1)} \\
		O_{k_1} &  (M^{j-2}-M^{i-2})^{(k_1)} \\
	\end{bmatrix}_{2k_1\times n}.
\end{align*}
By Lemma \ref{L4.1}, $(M^{j-2} - M^{i-2})$ is a nonzero element in the field $\mathbb{F}_q[M]$. Hence, $\mathrm{Rank}((M^{j-2} - M^{i-2})^{(k_1)}) = k_1$, which leads to the desired conclusion.

(3) By calculation, we obtain
\begin{align*}
	\mathbb{F}_q^n \setminus \{\oplus_{i=1}^{q^{k_2}+1} \mathcal{A}[i]\} =
	\text{rowsp}\begin{bmatrix}
		O_{k_2\times k_1} & I_{k_2}
	\end{bmatrix}_{k_2\times n} \setminus
	\text{rowsp}\begin{bmatrix}
		O_{k_1} & I_{k_1} & O_{k_1\times r}\\
	\end{bmatrix}_{k_1\times n} .
\end{align*}
Furthermore, the vectors in the set
\begin{align*}
	\{0\}\cup\text{rowsp}\begin{bmatrix}
		O_{k_2\times k_1} & I_{k_2}
	\end{bmatrix}_{k_2\times n} \setminus
	\text{rowsp}
	\begin{bmatrix}
		O_{k_1} & I_{k_1} & O_{k_1\times r}
	\end{bmatrix}_{k_1\times n}
\end{align*}
can span a subspace of $\mathbb{F}_q^n$ with dimension no greater than $r$. Given that $0 \leqslant r < k_1$, there exists no subspace of dimension $k_1$ within the set $(\mathbb{F}_q^n \setminus \{\oplus_{i=1}^{q^{k_2}+1} \mathcal{A}[i]\})\cup\{0\}$, which completes the proof.
\end{proof}

Next, we proceed to construct the second layer of the \textquotedblleft sandwich" structure.

\begin{proposition}\label{P4.5}
With the setup of Proposition \ref{P4.4}, define the matrices
\begin{equation}
	\begin{aligned}\label{LA3.3.1}
		B[1]:&=\begin{bmatrix}
			O_{r\times 2k_1} & I_{r}
		\end{bmatrix}_{r\times n},  \\
		B[2]:&=\begin{bmatrix}
			    O_{r\times k_1} & B
		\end{bmatrix}_{r\times n}, \\
		B[i]:&=\begin{bmatrix}
			O_{r\times k_1} & (M^{i-2})_{(r)}
		\end{bmatrix}_{r\times n},\ \ \
		i=3,\dots,q^{k_2}+1,
	\end{aligned}
\end{equation}
where $B$ is an $r \times k_2$ matrix of rank $r$, and $(M^i)_{(r)}$ represents the last $r$ rows of the $i$-th power of $M$. All such matrices form a set denoted by
\begin{align}\label{LA3.3.2}
	\mathbf{B}:= \{B[i]\mid i=1,\dots,q^{k_2}+1\}.
\end{align}
Furthermore, the rank of $B[i]$ is equal to $r$ for all $i = 1, \dots, q^{k_2}+1$.
\end{proposition}

\begin{proof}
According to the definition of companion matrix $M$ and $\text{Rank}(B) = r$, the result follows directly.	
\end{proof}

\begin{remark}\label{R4.6}
Here we present two remarks regarding the set $\mathbf{B}$.
\begin{itemize}
\item [(1)] If $r = 0$, the second layer $\mathbf{B}$ does not exist and the sandwich structure reduces to a double-decker configuration.
\item [(2)] The submatrix $B$ in $B[2]$ varies with the integer $r$, as will be demonstrated subsequently.
\end{itemize}
\end{remark}

We are now in a position to combine the sets $\mathbf{A}$ and $\mathbf{B}$ into a sandwich-like structure, thereby obtaining a class of constructions of full flag codes that include ODFC, QODFC and flag codes with other possible distance values.

\begin{definition}\label{D4.7}
With the setups of Proposition \ref{P4.4} and \ref{P4.5}, define the matrices
\begin{align}
	S[i]:=\begin{pmatrix}
		A[i]\\
		B[i]\\
		A[i+1]
	\end{pmatrix}_{n\times n},\quad  i=1,\dots,q^{k_2},\quad and\quad
	S[q^{k_2}+1]:=\begin{pmatrix}
		A[q^{k_2}+1]\\
		B[q^{k_2}+1]\\
		A[1]
\end{pmatrix}_{n\times n}.
\end{align}
These matrices form a set denoted by
\begin{align}
	\mathbf{S}:=\{S[i]\mid i=1,\dots,q^{k_2}+1\}.
\end{align}
Define the subspaces $\mathcal{F}_j^i := \mathrm{rowsp}(S[i]^{(j)})$ for $i = 1, 2, \dots, q^{k_2}+1$ and $j = 1, 2, \dots, n-1$. Collectively, define the set
\begin{align}
	\mathcal{C}:=\{\mathcal{F}^i := \{\mathcal{F}_1^i, \mathcal{F}_2^i, \dots, \mathcal{F}_{n-1}^i\}\mid i= 1,2,\dots,q^{k_2}+1\ \}.
\end{align}

\end{definition}

\begin{remark}\label{R4.8}
In the sequel, we aim to demonstrate that the rank of every matrix $S[i]$ within the set $\mathbf{S}$ is equal to $n$, and that the set $\mathcal{C}$ constitutes a full flag code.
\end{remark}

We are now prepared to construct a class of full flag codes. As noted in Remark \ref{R4.6}, the submatrix $B$ within $B[2]$ varies based on the value of  the integer $r$. Consequently, it is necessary to classify the cases of $r$.
\\

\noindent \textbf{Case 1: $r=0$}.

We proceed to introduce the first construction of full flag code.

If $r=0$, the set $\mathbf{S}$  presented in Definition \ref{D4.7} was originally introduced in \cite{ANS1}. Furthermore, the full flag code constructed based on the set $\mathbf{S}$ was proven to be an ODFC in the same reference. Herein, we demonstrate that this corresponds precisely to the case where $r = 0$ in Definition \ref{D4.7}.

\begin{lemma}\rm\cite{ANS1}\label{L4.9}
Let $r = 0$ and $n = k_1 + k_2 = 2k_1$. Then all matrices $S[i]$ in the set $\mathbf{S}$ become
\begin{align}
	S[i]:=\begin{pmatrix}
		A[i]\\
		A[i+1]
	\end{pmatrix}_{n\times n},\quad i=1,\dots,q^{k_1},\quad and \quad
	S[q^{k_1}+1]:=\begin{pmatrix}
		A[q^{k_1}+1]\\
		A[1]
	\end{pmatrix}_{n\times n},
\end{align}
and satisfy $\mathrm{Rank}(S[i])=n$.
\end{lemma}

\begin{proposition}\rm\cite [Theorem 4.5]{ANS1}
	Based on the setup provided in Lemma \ref{L4.9} and Definition \ref{D4.7}, the set $\mathcal{C}$ constitutes an optimal distance flag code with distance $d_f(\mathcal{C})=\frac{n^2}{2}=2k_1^2$ and achieves the best possible cardinality $|\mathcal{C}| = |\mathbf{S}| = q^{k_1} + 1$.
\end{proposition}
\hspace*{\fill}

\noindent \textbf{Case 2: $r=1$}.

We now turn to the second construction . Similar to Case 1, we shall first demonstrate that each matrix $S[i]$ is of full rank.

\begin{lemma}\label{L4.11}
Let $k_2 = k_1 + 1$ and $n = k_1 + k_2 = 2k_1 + 1$. Define the submatrix $B := \begin{pmatrix} 1 & 0 & 0 & \cdots & 0 \end{pmatrix}_{1 \times n}$ within $B[2]$ in Proposition \ref{P4.5}. Then the rank of each matrix $S[i]$ in the set $\mathbf{S}$ is equal to $n$.
\end{lemma}

\begin{proof}
Let us first compute $S[1]$ and $S[2]$. By Definition \ref{D4.7}, both
\begin{align*}
	 	S[1]=\begin{pmatrix}
	 		O_{k_1}          & I_{k_1} & O_{k_1\times 1}\\
	 		O_{1\times k_1}  &  O_{1\times k_1}  & 1    \\
	 		I_{k_1}          &  O_{k_1}     & O_{k_1\times 1}
	 	\end{pmatrix}_{n\times n}
 \end{align*}
and
\begin{align*}
S[2]=\left(\begin{array}{c|c}
	 		I_{k_1}         &  \begin{array}{cc} O_{k_1}   & O_{k_1\times 1} \end{array} \\
	\hline  O_{1\times k_1} &  \begin{array}{cc} 1\ \ \  & O_{1\times k_1} \end{array} \\
	\hline  I_{k_1}         &  M^{(k_1)}
	  \end{array}\right)_{n\times n}
	 \end{align*}
are full-rank matrices.
	
	 Now consider $i\geqslant 3$. By performing calculations, we have
\begin{align*}
S[i]=\begin{pmatrix}
      I_{k_1}         & (M^{i-2})^{(k_1)}\\
      O_{1\times k_1} & (M^{i-2})_{(1)}      \\
      I_{k_1}         & (M^{i-1})^{(k_1)}
\end{pmatrix}_{n\times n}
\rightarrow \begin{pmatrix}
	         I_{k_1}         & (M^{i-2})^{(k_1)}\\
	         O_{1\times k_1} & (M^{i-2})_{(1)}      \\
	         O_{k_1}         & ((M-I)M^{i-2})^{(k_1)}
           \end{pmatrix}_{n\times n} .
 \end{align*}
Furthermore, given that
\begin{align*}
\begin{pmatrix}
 (M^{i-2})_{(1)}      \\
 ((M-I)M^{i-2})^{(k_1)}
\end{pmatrix}
=\begin{pmatrix}
  I_{(1)}      \\
  (M-I)^{(k_1)}
 \end{pmatrix}\cdot M^{i-2}
=\begin{pmatrix}
 0 &  0 &      &\cdots&      & 1  \\
-1 &  1 &   0  &      &      &    \\
   & -1 &   1  &   0  &      &    \\
   &    &\ddots&\ddots&\ddots&    \\
   &    &      &  -1  & 1    & 0  \\
   &    &      &      & -1   & 1	 		
 \end{pmatrix}_{k_2\times k_2} \cdot M^{i-2},
\end{align*}
and $\text{Rank}(M)=k_2=k_1+1$, it is evident that $\textrm{Rank}(S[i])=n$ for any $i\in\{3, \dots,q^{k_2}+1\}$.
\end{proof}

Next, we present the specific construction of the full flag code when $r=1$.

\begin{theorem}\label{T4.12}
With the framework established by Lemma \ref{L4.11} and Definition \ref{D4.7}, the set $\mathcal{C}$ constitutes an optimal distance flag code, achieving the maximum distance $d_f(\mathcal{C})=2k_1k_2=\frac{n^2-1}{2}$ and the maximum possible cardinality $|\mathcal{C}| = |\mathbf{S}| = q^{k_2} + 1$.	
\end{theorem}
\begin{proof}
Since $n=k_1+k_2=2k_1+1$, we see that $t_L=k_1$ and $t_R=k_1+1=k_2$ according to Definition \ref{D2.7}. Consequently, the $L-th$ and $R-th$ projected codes of $\mathcal{C}$ are
\begin{align*}
	\mathcal{C}_L=\mathcal{C}_{k_1}=\{\mathcal{F}^i_{k_1}\mid\ i=1,\dots,q^{k_2}+1\}
\end{align*}
and
\begin{align*}
	\mathcal{C}_R=\mathcal{C}_{k_2}=\{\mathcal{F}^i_{k_2}\mid\ i=1,\dots,q^{k_2}+1\}.
\end{align*}

To prove that $\mathcal{C}$ is an ODFC, we only need to show that both $\mathcal{C}_L$ and $\mathcal{C}_R$ have maximum distance $2k_1$ and cardinality $|\mathcal{C}_L|=|\mathcal{C}_R|=|\mathcal{C}|$ by Lemma \ref{L2.4}.
Since every $\mathcal{F}^i_{k_1}=\mathrm{rowsp}(S[i]^{(k_1)})=\mathcal{A}[i]$, and the set $\mathfrak{A}= \{\mathcal{A}[i]\mid i=1,\dots,q^{k_2}+1\}$ constitutes a partial $k_1$-spread, it follows that $\mathcal{C}_L$ is a constant dimension code with the maximum distance $2k_1$, possessing a cardinality of $|\mathcal{C}_L| =|\mathcal{C}|=q^{k_2}+1$.

We next demonstrate that the distance of $\mathcal{C}_R = \mathcal{C}_{k_2}$ is $2k_1$. It suffices to show that $\dim(\mathcal{F}^i_{k_2}+\mathcal{F}^j_{k_2})=n$ for any $\mathcal{F}^i_{k_2} \neq \mathcal{F}^j_{k_2}$. This conclusion can be established by calculating the rank of the generated matrix of $\mathcal{F}^i_{k_2}+\mathcal{F}^j_{k_2}$.

Based on Definition \ref{D4.7}, $\mathcal{F}^i_{k_2} = \mathrm{rowsp}(S[i]^{(k_2)})$ for all $i = 1, \dots, q^{k_2 }+ 1$. It can be concluded that
\begin{align*}
\mathrm{Rank}\begin{pmatrix}
		S[1]^{(k_2)}\\
		S[2]^{(k_2)}
	    \end{pmatrix}
&=\mathrm{Rank}\left(\begin{array}{c|cc}
		O_{k_1}         & I_{k_1}         & O_{k_1\times 1} \\
		O_{1\times k_1} & O_{1\times k_1} & 1               \\
 \hline I_{k_1}         & O_{k_1}         & O_{k_1\times 1} \\
		O_{1\times k_1} &  1              & O_{1\times k_1}
        \end{array}\right)_{2k_2\times n} \\
&=\mathrm{Rank}\left( \begin{array}{c|cc}
		O_{k_1}         & I_{k_1}         & O_{k_1\times 1} \\
		O_{1\times k_1} & O_{1\times k_1} & 1               \\
 \hline I_{k_1}         & O_{k_1}         & O_{k_1\times 1} \\
	    \end{array}\right)_{n\times n}=n,\\
\textrm{and}\quad \mathrm{Rank}\begin{pmatrix}
		S[1]^{(k_2)}\\
		S[i]^{(k_2)}
	    \end{pmatrix}
&=\mathrm{Rank}\left( \begin{array}{c|c}
	   O_{k_1}         & \begin{array}{cc} I_{k_1} & O_{k_1\times 1}\end{array}\\
	   O_{1\times k_1} & \begin{array}{cc} O_{1\times k_1}\  & 1 \ \ \end{array} \\
\hline \begin{array}{c} I_{k_1}\\O_{1\times k_1} \end{array} & M^{i-2}
       \end{array}\right)_{2k_2\times n}\\	
&=\mathrm{Rank}\left( \begin{array}{c|c}
		O_{k_1}         & \begin{array}{cc} I_{k_1} & O_{k_1\times 1}\end{array}\\
		O_{1\times k_1} & \begin{array}{cc} O_{1\times k_1} & 1 \  \end{array} \\
 \hline I_{k_1}         &   (M^{i-2})^{(k_1)}	
	    \end{array}\right)_{n\times n}=n \quad \mathrm{for}\ i\geqslant 3.
\end{align*}
Furthermore, for $i,j\in\{2,\dots,q^{k_2}+1\}$ and $i\neq j$, it yields that
\begin{align*}
	\left( \begin{array}{c}
		S[i]^{(k_2)}\\
		S[j]^{(k_2)}
	\end{array}\right) =
	\left( \begin{array}{c|c}
		\begin{array}{c} I_{k_1}\\O_{1\times k_1} \end{array} & M^{i-2}\\
		\hline\begin{array}{c} I_{k_1}\\O_{1\times k_1} \end{array} & M^{j-2}
	\end{array}\right)_{2k_2\times n} \rightarrow
	\left(\begin{array}{c|c}
		\begin{array}{c} I_{k_1}\\O_{1\times k_1} \end{array} & M^{i-2}\\
		\hline O_{k_2\times k_1}                                   & M^{j-2}-M^{i-2}
	\end{array} \right)_{2k_2\times n}   ,
\end{align*}
and
\begin{align*}
\mathrm{Rank}
	\left(\begin{array}{c|c}
		\begin{array}{c} I_{k_1}\\O_{1\times k_1} \end{array} & M^{i-2}\\
		\hline O_{k_2\times k_1}                                   & M^{j-2}-M^{i-2}
	\end{array} \right)_{2k_2\times n}  =
\mathrm{Rank}
	\left(\begin{array}{c|c}
		\begin{array}{c} I_{k_1} \end{array} & (M^{i-2})^{(k_1)}\\
		\hline O_{k_2\times k_1}                                   & M^{j-2}-M^{i-2}
	\end{array} \right)_{n\times n}=n.
\end{align*}
Consequently, as derived from the distance (\ref{F2.1}), it is evident that the $R$-th projected code $\mathcal{C}_R$ constitutes a constant dimension code with maximum distance $d_S(\mathcal{C}_R)=2(n-k_2) = 2k_1$, and cardinality $|\mathcal{C}_R| = |\mathcal{C}| = q^{k_2} + 1$.

Therefore, $\mathcal{C}$ constitutes an optimal distance flag code with distance $d_f(\mathcal{C})=2k_1k_2=\frac{n^2-r^2}{2}$ and achieves the maximum possible cardinality $|\mathcal{C}| = |\mathbf{S}| = q^{k_2} + 1$.
\end{proof}

\begin{remark}
 In \cite{K}  S. Kurz provides a construction of full flag code which has the same cardinality and minimum distance as the one showed in Theorem \ref{T4.12}, but no decoding algorithm is proposed. One contribution of this paper is introducing a well-structured construction for this code in terms of block matrices. Thanks to this
description, in Section 5 we are able to provide an efficient decoding algorithm for such code.
\end{remark}

\hspace*{\fill}

\noindent \textbf{Case 3: $r\geqslant 2$}.

We proceed to show the third construction for the full flag code. In particular, this full flag code is precisely the QODFC for $r = 2$. 

First, we present the submatrix $B$ in $B[2]$ from the set $\mathbf{B}$.

\begin{lemma}\label{L4.13}
With the setup of Definition \ref{D4.7}, let $k_2=k_1+r$, $n=k_1+k_2=2k_1+r$, and the submatrix 
\begin{align*}
B:=\begin{pmatrix}
	\begin{array}{cccccc} 1 & 0 & 0 & \cdots & 0 & 0\end{array}\\
	\begin{array}{cc} O_{(r-1)\times(k_1+1)} & I_{(r-1)}\end{array}
\end{pmatrix}_{r\times k_2}
\end{align*}
in $B[2]$.
Then the rank of every matrix $S[i]$ within the set $\mathbf{S}$ is equal to $n$.	
\end{lemma}

\begin{proof}
The proof can be handled in much the same way as Lemma \ref{L4.11}, the only difference being in the analysis of the matrix $B[2]$. By culculation,
 \begin{align*}
S[2]=\left(\begin{array}{c|c}
		I_{k_1}         &  O_{k_1\times k_2} \\
 \hline O_{r\times k_1} & \begin{array}{c}
			              \begin{array}{cccccc} 1 & 0 & 0 & \cdots & 0 & 0 \end{array}\\
			              \begin{array}{cc} O_{(r-1)\times(k_1+1)} & I_{(r-1)}\end{array}
		                  \end{array}\\
 \hline I_{k_1}         &  M^{(k_1)}
	  \end{array}\right)_{n\times n} 
\rightarrow
      \left(\begin{array}{c|c}
		I_{k_1}         &  O_{k_1\times k_2}           \\
 \hline O_{r\times k_1} & \begin{array}{c}
 	                      \begin{array}{cccccc} 1 & 0 & 0 & \cdots & 0 & 0 \end{array}\\
 	                      \begin{array}{cc} O_{(r-1)\times(k_1+1)} & I_{(r-1)}\end{array}
                          \end{array}\\
 \hline O_{k_1}         & \begin{array}{ccc} O_{k_1\times1} & I_{k_1} & O_{k_1\times (r-1)}\end{array}
	\end{array}\right)_{n\times n}
\end{align*}
is a full-rank matrix. 
\end{proof}

In what follows, we present the cardinality and distance of the flag code $\mathcal{C}$. Before proceeding with these details, we first provide a lemma that will be utilized in the calculation of the distance.

\begin{lemma}\label{L4.14}
With the framework established by Lemma \ref{L4.13} and Definition \ref{D4.7}, the distance of the projected code $\mathcal{C}_l$ is equal to $2k_1$ for all $l = k_1 + 1, k_1 + 2, \ldots, k_2 - 1$.
\end{lemma}

\begin{proof}
The proof is by induction on the integer $l$.

If $l=k_1+1$, it yields that $\mathcal{F}^i_{k_1}\subset\mathcal{F}^i_{k_1+1}$ and that  $\dim(\mathcal{F}^i_{k_1+1})=k_1+1$ for all $i=1,\dots,q^{k_2}+1$. For any two distinct $(k_1+1)$-dimensional subspaces $\mathcal{F}^i_{k_1+1},\mathcal{F}^j_{k_1+1}\in \mathcal{C}_{k_1+1}$, we have
\begin{align*}
	 d_S(\mathcal{F}^i_{k_1+1},\mathcal{F}^j_{k_1+1})=2[(k_1+1)-\dim(\mathcal{F}^i_{k_1+1}\cap\mathcal{F}^j_{k_1+1})].
\end{align*}
 Since $\mathcal{C}_{k_1}$ is a partial spread, it follows that
\begin{align*}
	 0=\dim(\mathcal{F}^i_{k_1}\cap\mathcal{F}^j_{k_1})\leqslant\dim(\mathcal{F}^i_{k_1+1}\cap\mathcal{F}^j_{k_1+1})\leqslant2.
\end{align*}

According to Proposition \ref{P4.5}, the $(k_1+1)$-th row of the generator matrix $S[i]$ (resp.$S[j]$) is precisely the vector $B[i]^{(1)}$ (resp.$B[j]^{(1)}$). Hence
\begin{align*}
B[i]^{(1)}\in \mathcal{A}[j]=\mathrm{rowsp}(S[j]^{(k_1)})\ \textrm{and} \ B[j]^{(1)}\in \mathcal{A}[i]=\mathrm{rowsp}(S[i]^{(k_1)})
\end{align*}
 can not both hold true by Proposition \ref{P4.4}, i.e., $\dim(\mathcal{F}^i_{k_1+1}\cap\mathcal{F}^j_{k_1+1})<2$. Therefore, we have the following equations:
\begin{align*}
	&d_S(\mathcal{F}^i_{k_1+1},\mathcal{F}^j_{k_1+1})=2(k_1+1-1)=2k_1\\
	\text{or}\ \ &d_S(\mathcal{F}^i_{k_1+1},\mathcal{F}^j_{k_1+1})=2(k_1+1-0)=2k_1+2.
\end{align*}
That is, we can express this as:
\begin{align*}
	d_S(\mathcal{C}_{k_1+1})=2k_1\ \mathrm{or}\ d_S(\mathcal{C}_{k_1+1})=2k_1+2.
\end{align*}
Now assume that $d_S(\mathcal{C}_{k_1+1})=2k_1+2$, then we have $\mathcal{F}^i_{k_1+1}\cap\mathcal{F}^j_{k_1+1}=\{0\}$ for any two distinct $\mathcal{F}^i_{k_1+1},\mathcal{F}^j_{k_1+1}\in \mathcal{C}_{k_1+1}$. Consequently, the projected code $\mathcal{C}_{k_1+1}$ forms a partial $(k_1+1)$-spread, with $|\mathcal{C}_{k_1+1}|=q^{k_2}+1$. However, the maximum size of a partial $(k_1+1)$-spread is
\begin{align*}
	 \lfloor\dfrac{q^n-1}{q^{k_1+1}-1}\rfloor&=\dfrac{q^{r-2}(q^{2k_1+2}-1)}{q^{k_1+1}-1}+\lfloor\dfrac{q^{r-2}-1}{q^{k_1+1}-1}\rfloor \\ \nonumber
	&=q^{r-2}(q^{k_1+1}+1)\\ \nonumber
	&=q^{k_2-1}+q^{r-2}\\
	&<q^{k_2}+1=|\mathcal{C}_{k_1+1}|
\end{align*}
which presents a contradiction. Consequently, we conclude that $d_S(\mathcal{C}_l)=d_S(\mathcal{C}_{k_1+1})=2k_1$.

Next, assume that $d_S(\mathcal{C}_l)=2k_1$ holds for the integer $l$, and we will prove it for $l+1$.

For every pair of distinct subspaces $\mathcal{F}^i_{l+1},\mathcal{F}^j_{l+1}\in\mathcal{C}_{l+1}$, the corresponding subspaces $\mathcal{F}^i_l$ and $\mathcal{F}^j_l$ represent two distinct subspaces within $\mathcal{C}_l$. Given that $d_S(\mathcal{C}_l) = 2k_1$, it follows that
\begin{align*}
	d_S(\mathcal{F}^i_l,\mathcal{F}^j_l)=2(\dim(\mathcal{F}^i_l+\mathcal{F}^j_l)-l)\geqslant2k_1.
\end{align*}
Notice that there exists a pair of distinct subspaces such that the equality holds since $d_S(\mathcal{C}_l) = 2k_1$.

We next examine this topic by considering two distinct scenarios.

(1) When $d_S(\mathcal{F}^i_l,\mathcal{F}^j_l)=2k_1$:

It follows that $\dim(\mathcal{F}^i_l+\mathcal{F}^j_l)=k_1+l$. In the subsequent discussion, we aim to demonstrate that $\dim(\mathcal{F}^i_{l+1}+\mathcal{F}^j_{l+1})=k_1+l+1$.

Since the subspaces $\mathcal{F}^i_l=\mathrm{rowsp}(S[i]^{(l)})$ and $\mathcal{F}^j_l=\mathrm{rowsp}(S[j]^{(l)})$, the matrix
\begin{align*}
	\begin{pmatrix}
		S[i]^{(l)}\\
		S[j]^{(l)}
	\end{pmatrix} &=
	\left(\begin{array}{c|c}
		\begin{array}{c} I_{k_1}\\O_{(l-k_1)\times k_1} \end{array} & (M^{i-2})^{(l)} \\
		\hline \begin{array}{c} I_{k_1}\\O_{(l-k_1)\times k_1} \end{array} & (M^{j-2})^{(l)}
	\end{array}\right)_{2l\times n} \\
&\rightarrow
	\left(\begin{array}{c|c}
		\begin{array}{c} I_{k_1}\\O_{(l-k_1)\times k_1} \end{array} & (M^{i-2})^{(l)} \\
		\hline O_{l\times k_1} & (M^{j-2}-M^{i-2})^{(l)}
	\end{array}\right)_{2l\times n}.
\end{align*}
According to Lemma \ref{L4.1}, define $M^\nu:=(M^{j-2}-M^{i-2})$  for some $\nu\in\{1,\dots,q^{k_2}-1\}$. Since $\dim(\mathcal{F}^i_l+\mathcal{F}^j_l)=k_1+l$, it is essential to ensure that
\begin{equation}
	\begin{aligned}\label{F3.3}
\mathrm{Rank}\begin{pmatrix}
	S[i]^{(l)}\\
	S[j]^{(l)}
    \end{pmatrix}
&=\mathrm{Rank}\left(\begin{array}{c|c}
	     \begin{array}{c} I_{k_1}\\O_{(l-k_1)\times k_1} \end{array} & (M^{i-2})^{(l)} \\
  \hline O_{l\times k_1} & (M^\nu)^{(l)}
         \end{array}\right)_{2l\times n}\\
&=k_1+l\\
&=\mathrm{Rank}\left(\begin{array}{c|c}
	     I_{k_1}  & (M^{i-2})^{(k_1)} \\
  \hline O_{l\times k_1} & (M^\nu)^{(l)}
         \end{array}\right)_{(l+k_1)\times n}.
\end{aligned}
\end{equation}
By Lemma \ref{L4.2} (2), let the last $(l-k_1)$ rows of $(M^{i-2})^{(l)}$ represented as $v_0,v_0M,\dots,v_0M^{l-k_1-1},$ respectively, and let the first $l$ rows of $M^\nu$ be denoted by $v_1,v_1M,\dots,v_1M^{l-1},$ respectively. Then from (\ref{F3.3}), every vector in the set $\{v_0,v_0M,\dots,v_0M^{l-k_1-1}\}$ can be expressed as a linear combination of the vectors $v_1,v_1M,\dots,v_1M^{l-1}$. Consequently, it follows that $v_0M^{l-k_1}$ is also a linear combination of the vectors $v_1M,v_1M^2,\dots,v_1M^{l}$. 

Calculating the generated matrix of $\mathcal{F}^i_{l+1}+\mathcal{F}^j_{l+1}$, it yields that
\begin{align*}
\begin{pmatrix}
	S[i]^{(l+1)}\\
	S[j]^{(l+1)}
\end{pmatrix}
&=\left(\begin{array}{c|c}
	    \begin{array}{c} I_{k_1}\\O_{(l-k_1+1)\times k_1} \end{array} & (M^{i-2})^{(l+1)} \\
  \hline\begin{array}{c} I_{k_1}\\O_{(l-k_1+1)\times k_1} \end{array} & (M^{j-2})^{(l+1)}
        \end{array}\right)_{2(l+1)\times n} \\
&\rightarrow
\left(\begin{array}{c|c}	
		 \begin{array}{c} I_{k_1}\\O_{(l-k_1+1)\times k_1} \end{array} & (M^{i-2})^{(l+1)} \\
  \hline O_{(l+1)\times k_1} & (M^\nu)^{(l+1)}
	     \end{array}\right)_{2(l+1)\times n}.
\end{align*}
It is evident that the last $(l-k_1+1)$ rows of $(M^{i-2})^{(l+1)}$ are precisely $v_0,v_0M,\dots,v_0M^{l-k_1},$ respectively, which represent linear combinations of $v_1,v_1M,\dots,v_1M^l$. Hence,
\begin{align*}
\mathrm{Rank}\begin{pmatrix}
	S[i]^{(l+1)}\\
	S[j]^{(l+1)}
    \end{pmatrix}
&=\mathrm{Rank}\left(\begin{array}{c|c}
	     \begin{array}{c} I_{k_1}\\O_{(l-k_1+1)\times k_1} \end{array} & (M^{i-2})^{(l+1)} \\
  \hline O_{(l+1)\times k_1} & (M^\nu)^{(l+1)}
         \end{array}\right)_{2(l+1)\times n}\\
&=\mathrm{Rank}\left(\begin{array}{c|c}
	     I_{k_1}  & (M^{i-2})^{(k_1)} \\
  \hline O_{(l+1)\times k_1} & (M^\nu)^{(l+1)}
         \end{array}\right)_{(k_1+l+1)\times n}\\
&=k_1+l+1	
\end{align*}
and $\dim(\mathcal{F}^i_{l+1}+\mathcal{F}^j_{l+1})=k_1+l+1$. Therefore the distance between the subspaces is given by $d_S(\mathcal{F}^i_{l+1},\mathcal{F}^j_{l+1})=2(k_1+l+1-(l+1))=2k_1$.

(2) When $d_S(\mathcal{F}^i_l,\mathcal{F}^j_l)>2k_1$:

The statement implies that $\dim(\mathcal{F}^i_l+\mathcal{F}^j_l)>k_1+l$. Given that $\mathcal{F}^i_l\subset\mathcal{F}^i_{l+1}$ and $\mathcal{F}^j_l\subset\mathcal{F}^j_{l+1}$, it follows that
\begin{align*}
\dim(\mathcal{F}^i_{l+1}+\mathcal{F}^j_{l+1})\geqslant \dim(\mathcal{F}^i_l+\mathcal{F}^j_l) \geqslant k_1+l+1>k_1+l.
\end{align*} 
Consequently, we have
\begin{align*}
d_S(\mathcal{F}^i_{l+1},\mathcal{F}^j_{l+1})\geqslant 2(k_1+l+1-(l+1))=2k_1.
\end{align*}

In conclusion, it is established that $d_S(\mathcal{C}_{l+1})=2k_1$.
\end{proof}

We are now in position to calculate the distance and cardinality of $\mathcal{C}$.

\begin{theorem}\label{T4.15}
	With the framework established by Lemma \ref{L4.13} and Definition \ref{D4.7}, the set $\mathcal{C}$ constitutes a full flag code with distance $d_f(\mathcal{C}) = \frac{n^2 - r^2}{2}$ and cardinality $|\mathcal{C}| = |\mathbf{S}| = q^{k_2} + 1$.
\end{theorem}

\begin{proof}
Since $\mathrm{Rank}(S[i])=n$, then $\dim(\mathcal{F}^i_j)=\mathrm{Rank}(S[i]^{(j)})=j$, and  $\mathcal{F}^i=\{\mathcal{F}_1^i,\mathcal{F}_2^i,\dots,\mathcal{F}_{n-1}^i\}$ constitutes a full flag for all $i=1,\dots,q^{k_2}+1$ and $j=1,\dots,n-1$. Hence, the set $\mathcal{C}$ represents a full flag code with cardinality $|\mathcal{C}| =|\mathbf{S}|=q^{k_2}+1$.

Now let us calculate the distance of each projected code  $\mathcal{C}_i$.

Firstly, since the $k_1$-th projected code $\mathcal{C}_{k_1}$  constitutes a partial $k_1$-spread of $\mathbb{F}_q^n$ by Proposition \ref{P4.4}, the distance $d_S(\mathcal{C}_{k_1})=2k_1$. Furthermore, according to the subspaces inclusion relation $\mathcal{F}^i_1\subset\mathcal{F}^i_2\subset\cdots\subset\mathcal{F}^i_{k_1}$, it is established that $d_S(\mathcal{C}_l)=2l$  for all $l=1,\cdots,k_1$.

In a manner analogous to the proof presented in Theorem \ref{T4.12}, we can deduce that $d_S(\mathcal{C}_{k_2})=2(n-k_2)=2k_1$. Additionally, according to the nested structure of subspaces $\mathcal{F}^i_{k_2}\subset\mathcal{F}^i_{k_2+1}\subset\cdots\subset\mathcal{F}^i_{n-1}$, it follows that $d_S(\mathcal{C}_l)=2(n-l)$ for all $l=k_2,\dots,n-1$.

Besides, based on Lemma \ref{L4.14}, it holds true that $d_S(\mathcal{C}_l)=2k_1$ for all $l=k_1+1,k_1+2,\dots,k_2-1$. Hence, we obtain the distance of the full flag code
\begin{align*}
	d_f(\mathcal{C})\geqslant\sum_{l=1}^{n-1}d_S(\mathcal{C}_l)=2\sum_{l=1}^{k_1}l+2k_1\times (r-1)+2\sum_{l=k_2}^{n-1}(n-l)=\frac{n^2-r^2}{2}.
\end{align*}	
Neverthless, based on 
\begin{align*}
S[1]=\begin{pmatrix}
	O_{k_1}          & I_{k_1} & O_{k_1\times r}\\
	O_{r\times k_1}  &  O_{r\times k_1}  & I_r    \\
	I_{k_1}          &  O_{k_1}     & O_{k_1\times r}
\end{pmatrix}_{n\times n}
\mathrm{and}\quad
S[2]=\left(\begin{array}{c|c}
	I_{k_1}         &  O_{k_1\times k_2} \\
	\hline O_{r\times k_1} & \begin{array}{c}
		\begin{array}{cccccc} 1 & 0 & 0 & \cdots & 0 & 0 \end{array}\\
		\begin{array}{cc} O_{(r-1)\times(k_1+1)} & I_{(r-1)}\end{array}
	\end{array}\\
	\hline I_{k_1}         &  M^{(k_1)}
\end{array}\right)_{n\times n},
\end{align*}
 it yields that $d_f(\mathcal{F}^1,\mathcal{F}^2)=\frac{n^2-r^2}{2}\geqslant d_f(\mathcal{C})$. In conclution, we have $d_f(\mathcal{C})=\frac{n^2-r^2}{2}$.
\end{proof}

According to the proof of Theorem \ref{T4.15}, we have drawn the following result:

\begin{corollary}\label{C4.18}
The distance of the full flag code $\mathcal{C}$ is $d_f(\mathcal{C})=\sum_{l=1}^{n-1}d_S(\mathcal{C}_l)$.
\end{corollary}

We shall now proceed to examine the case where $r=2$.

\begin{theorem}\label{T4.16}
 	If $r=2$, then the set $\mathcal{C}$ is a quasi-optimal distance flag code. Moreover, if $k_1 > q+1$, then $|\mathcal{C}|$ achieves its maximum value $A_q(n,2k_1;k_1)=q^{k_1+2}+1$.
\end{theorem}

\begin{proof}
Since $r=2$, the integer $n=2k_1+r$ is even.  According to Theorem \ref{T4.15} and Lemma \ref{L2.5}, the distance of $\mathcal{C}$ equals $d_f(\mathcal{C}) = (n^2-r^2)/2 = D^{(n)}-2$.  Consequently,  $\mathcal{C}$ constitutes a quasi-optimum distance flag code.

If $k_1 > \frac{q^2-1}{q-1}=q+1$, then the $k_1$-projected code $\mathcal{C}_{k_1}$, being a constant dimension code of dimension $k_1$, implies that
\begin{align*}
	q^{k_2}+1=|\mathcal{C}|=|\mathcal{C}_{k_1}|\leqslant A_q(n,2k_1;k_1)=\frac{q^{2k_1+2}-q^{k_1+2}}{q^{k_1}-1}+1=q^{k_1+2}+1.
\end{align*}	
Lemma \ref{L2.2} guarantees the validity of the penultimate equality.

Therefore, $|\mathcal{C}|$ attains its maximum value $A_q(n,2k_1;k_1)=q^{k_1+2}+1$.
\end{proof}

\begin{remark}\label{R3.18}
In \cite{AN3}, the authors proposed a construction method for QODFC with cardinality $q^{(n-2)/2}+1$ where $n$ is even. In Theorem \ref{T4.16}, we introduce an improved QODFC construction method with cardinality $|\mathcal{C}|=q^{k_1+2}+1=q^{(n+2)/2}+1$ which surpasses that in \cite{AN3}. Specifically, if $k_1 > \frac{q^2-1}{q-1}=q+1$, then $|\mathcal{C}|$ attains its maximum value. For the case where $k_1 \leqslant q+1$, Lemma \ref{L2.1} demonstrates that $|\mathcal{C}|=|\mathcal{C}_{k_1}|$ deviates from the theoretical maximum by at most $q^2-1$.
\end{remark}

When $r\geqslant3$, the distance $d_f(\mathcal{C})$ of the full flag code $\mathcal{C}$ deviates significantly from $D^{(n)}$, meaning that $\mathcal{C}$ is no longer an ODFC or QODFC. A construction for full flag codes with the third best distance and other distances was proposed in \cite{AN3}, where the cardinality is $q^{(n-1)/2} + 1$ (for $n$ is odd) or $q^{(n-2)/2} + 1$ (for $n$ is even). The constructions presented in this paper yield full flag codes with the same distances as those in \cite{AN3}, but achieve a greater cardinality. For example, when $r=3$,
\begin{align*}
|\mathcal{C}|=q^{k_1+r}+1=q^{\frac{n+3}{2}}+1>q^{\frac{n-1}{2}}+1;
\end{align*}
when $r=4$, 
\begin{align*}
|\mathcal{C}|=q^{k_1+r}+1=q^{\frac{n+4}{2}}+1>q^{\frac{n-2}{2}}+1.
\end{align*}
Although the cardinality of $\mathcal{C}$ for $r\geqslant3$ does not attain its theoretical upper bound, no constructions with a larger cardinality under the same distance condition have been found in existing studies.

\section{Decoding algorithm on erasure channel}

In this section, we present an efficient decoding algorithm for these full flag codes shown in the previous section. For convenience, assume that a full flag $\mathcal{F} = (\mathcal{F}_1,\dots,\mathcal{F}_{n-1})$ is sent, and the receiver obtain a sequence of subspaces $\mathcal{X} = (\mathcal{X}_1,\dots,\mathcal{X}_{n-1})$.

Two types of errors of flag codes are described in \cite{ANS1}: \textit{erasures} and \textit{insertions}. The former represent the lost information from each sent subspace $\mathcal{F}_i$ in the received one $\mathcal{X}_i$ and the latter represent the dimensions of the vector space generated by the set of vectors in $\mathcal{X}_i$ but not in $\mathcal{F}_i$. The specific definitions regarding the numbers of errors are provided below.

\begin{definition}\cite{ANS1}
Set $\mathcal{X}_i=\mathcal{\bar{F}}_i\oplus\mathcal{E}_i$ for every $i \in\{1,\dots,n-1\}$, where $\mathcal{\bar{F}}_i$ is a subspace of $\mathcal{F}_i$ and $\mathcal{F}_i\cap \mathcal{E}_i=\{0\}$.
The \textit{number of eraseres} is
\begin{align}
	d_S(\mathcal{F}_i,\mathcal{\bar{F}}_i)=\dim(\mathcal{F}_i)-\dim(\mathcal{\bar{F}}_i)
\end{align}
and the \textit{number of insertions} is $\dim(\mathcal{E}_i)$. The number of errors at the i-th shot of $\mathcal{F}$ is
\begin{align}
	e_i:=d_S(\mathcal{F}_i,\mathcal{X}_i)=d_S(\mathcal{F}_i,\mathcal{\bar{F}}_i)+\dim(\mathcal{E}_i).
\end{align}
The \textit{total numbers of errors} is given by
\begin{align}	
	e:=d_f(\mathcal{F},\mathcal{X})=\sum_{i=1}^{n-1} d_S(\mathcal{F}_i,\mathcal{X}_i)=\sum_{i=1}^{n-1} e_i.
\end{align}
\end{definition}

\begin{remark}
According to the coding theory, any code with a minimum distance $d$ is capable of detecting up to $(d - 1)$ errors and correcting up to $\lfloor\frac{d-1}{2}\rfloor$ errors.	
\end{remark}

As a premise, the network is modeled as a finite directed acyclic multigraph, featuring a single source and potentially multiple receivers. The source and the receivers agree upon a set of flags, with the information being encoded as these flags. Given that we are operating within an erasure channel framework, only erasures are permitted during the transmission process. We next divide the process into three steps, which was originally introduced in \cite{LNV}.

\begin{itemize}
\item[(1)] Through each outgoing edge, the source transmits the $i$-th row of the generator matrix corresponding to $\mathcal{F}_i$ for $i=1,\dots,n-1$.
\item[(2)] Every intermediate node forms a random linear combination of everything received up to this point for every edge.
\item[(3)] The receiver obtains $i$ random linear combinations, which generate a subspace denoted by $\mathcal{X}_i$. 
\end{itemize}

We now propose an enhanced decoding algorithm specifically designed for our construction, building upon the one presented in \cite{ANS1}. The algorithm will be divided into three distinct steps.

\noindent \textbf{Step 1:}

When at least one of the subspaces $\mathcal{X}_1, \dots, \mathcal{X}_{k}$ is non-trivial, we can uniquely decode $\mathcal{X}$ into the corresponding $\mathcal{F}$ such that $\mathcal{F}_i \in \mathcal{C}_i$.

\begin{proposition}\label{P5.1}
	If there exists an index $i\in\{1,\dots, k_1\}$ such that the corresponding received subspace $\mathcal{X}_i$ is non-trivial, then we decode $\mathcal{X}$ into the unique $\mathcal{F} = (\mathcal{F}_1,\dots,\mathcal{F}_i,\dots,\mathcal{F}_{n-1})$ such that $\mathcal{X}_i\subseteq\mathcal{F}_i$.
\end{proposition}
\begin{proof}
	In the erasure channel, $\mathcal{X}_i$ is a subspace of $\mathcal{F}_i$. Since the projected code $\mathcal{C}_i$ is a partial spread for every $i\in\{1,\dots, k_1\}$, $\mathcal{F}_i$ must be the only subspace in $\mathcal{C}_i$ that contains $\mathcal{X}_i$. Hence $\mathcal{F} = (\mathcal{F}_1,\dots,\mathcal{F}_i,\dots,\mathcal{F}_{n-1})$ is the only flag that contains $\mathcal{F}_i$ due to $|\mathcal{C}_i| = |\mathcal{C}|$.
\end{proof}

\noindent \textbf{Step 2:}

When all subspaces $\mathcal{X}_1=\cdots=\mathcal{X}_{k_1}=\{0\}$, and at least one of the subspaces $\mathcal{X}_{k_1+1}, \dots, \mathcal{X}_{k_1+r}$ is non-trivial, we cannot apply the algorithm outlined in Step 1 since none of the projected codes $\mathcal{C}_{k_1+1},\dots,\mathcal{C}_{k_1+r}$ is a partial spread. Therefore, we will present an alternative algorithm that utilizes the dimension of $\mathcal{X}_i$.

To minimize the number of errors in advance and enhance the informational content of the received codewords, we implement the method similar to the one shown in \cite{ANS1}. Set the subspaces
\begin{align}
	\mathcal{Y}_i=\{0\} \ \textrm{for}\ i=1,\dots,k_1,\ \textrm{and}\ \mathcal{Y}_i=\oplus^i_{j=k_1+1}\mathcal{X}_j \ \textrm{for}\ i=k_1+1,\dots,n-1.
\end{align}
So we have $\mathcal{Y}_1\subseteq \mathcal{Y}_2\subseteq \cdots \subseteq \mathcal{Y}_{n-1}$, $\mathcal{Y} = (\mathcal{Y}_1,\dots,\mathcal{Y}_{n-1})$ and $\mathcal{Y}_i\subseteq\mathcal{F}_i$ for all $i=1,\dots,n-1$.

\begin{proposition}
Assume that a subspace sequence $\mathcal{Y} = (\mathcal{Y}_1,\dots,\mathcal{Y}_{n-1})$ is received, where $\mathcal{Y}_1=\cdots=\mathcal{Y}_{k_1}=\{0\}$. If there exists an index $i\in\{k_1+1,\dots,k_1+r\}$, such that the subspace $\mathcal{Y}_i$ satisfies $\dim(\mathcal{Y}_i)> i-k_1$, then we decode $\mathcal{Y}$ into the unique $\mathcal{F}$ such that $\mathcal{Y}_i\subseteq\mathcal{F}_i$.
\end{proposition}
\begin{proof}
According to Lemma \ref{L4.14}, the distance $d_S(\mathcal{C}_i)=2k_1$ for all $i\in \{k_1+1,\dots,k_1+r\}$. Consequently, we have
\begin{align*}
\dim(\mathcal{F}_i+\mathcal{F}'_i)-\dim(\mathcal{F}_i\cap\mathcal{F}'_i)=2i-2\dim(\mathcal{F}_i\cap\mathcal{F}'_i)\geqslant d_S(\mathcal{C}_i)=2k_1,
\end{align*}
which implies that $\dim(\mathcal{F}_i\cap\mathcal{F}'_i)\leqslant i-k_1$. Given that $\dim(\mathcal{Y}_i)> i-k_1$, it follows that there exists an unique full flag $\mathcal{F}=\{\mathcal{F}_1,\dots,\mathcal{F}_i,\dots,\mathcal{F}_{n-1}\}$ such that $\mathcal{Y}_i\subseteq\mathcal{F}_i$.
\end{proof}

\noindent \textbf{Step 3:}
When all subspaces $\mathcal{Y}_1=\cdots=\mathcal{Y}_{k_1}=\{0\}$, and each of the subspaces $\mathcal{Y}_i$ satisfies $\dim(\mathcal{Y}_i)\leqslant i-k_1$ for $i\in\{k_1+1,\dots, k_1+r\}$, we again present an algorithm that utilizes the dimensions of $\mathcal{Y}_{k_1+r+1},\dots, \mathcal{Y}_{n-1}$. The following lemma serves as a foundation for our decoding algorithm.

\begin{lemma}\label{L5.5}
Assume a subspace sequence $\mathcal{Y} = (\mathcal{Y}_1,\dots,\mathcal{Y}_{n-1})$ is received, where the subspaces $\mathcal{Y}_1=\cdots=\mathcal{Y}_{k_1}=\{0\}$, and all the subspaces $\mathcal{Y}_i$ satisfy $\dim(\mathcal{Y}_i)\leqslant i-k_1$ for $i\in\{k_1+1,\dots, k_1+r\}$. If the total number of errors $e\leqslant\lfloor\frac{d_f(\mathcal{C})-1}{2}\rfloor$, then there exists an index  $i_0\in\{k_1+r+1,\dots,n-1\}$ such that $\dim(\mathcal{Y}_{i_0})>2i_0-n$.
\end{lemma}
\begin{proof}
Assume that $\dim(\mathcal{Y}_i)\leqslant2i-n$ for all $i\in\{k_1+r+1,\dots, n-1\}$. Since $\mathcal{Y}_i\subseteq\mathcal{F}_i$, it yields that
\begin{align*}
e_i=d_S(\mathcal{Y}_i,\mathcal{F}_i)=\dim(\mathcal{F}_i)-\dim(\mathcal{Y}_i)=i-\dim(\mathcal{Y}_i).
\end{align*}
 Hence
\begin{align*}
	e_i&=i=\frac{d_S(\mathcal{C}_i)}{2},\quad \mathrm{for}\quad i=1,\dots,k_1,\\
	e_i&=i-\dim(\mathcal{Y}_i)\geqslant k_1=\frac{d_S(\mathcal{C}_i)}{2},\quad \mathrm{for}\quad i=k_1+1,\dots,k_1+r,\\
	e_i&=i-\dim(\mathcal{Y}_i) \geqslant n-i=\frac{d_S(\mathcal{C}_i)}{2},\quad \mathrm{for}\quad i=k_1+r+1,\dots,n-1.
\end{align*}	
Then the total number of errors 
\begin{align*}
e=\sum\limits_{i=1}^{n-1} e_i \geqslant\sum\limits_{i=1}^{n-1}\frac{d_S(\mathcal{C}_i)}{2}=\frac{d_f(\mathcal{C})}{2} >\lfloor\frac{d_f(\mathcal{C})-1}{2}\rfloor
\end{align*}
where Corollary \ref{C4.18} guarantees the validity of the equality. And it cannot be correctable, which is a contradiction.
\end{proof}

Next, we present the specific content of the decoding algorithm.

\begin{proposition}
Let the total number of errors $e\leqslant\lfloor\frac{d_f(\mathcal{C})-1}{2}\rfloor$. Assume that a sequence of subspaces $\mathcal{Y}=(\mathcal{Y}_1,\dots,\mathcal{Y}_{n-1})$ is received with $\mathcal{Y}_1=\cdots=\mathcal{Y}_{k_1}=\{0\}$, and that each of the subspaces $\mathcal{Y}_i$ satisfies $\dim(\mathcal{Y}_i)\leqslant i-k_1$ for $i\in\{k_1+1,\dots, k_1+r\}$. Let $k_1+r< i_0 <n$ denote the minimum index such that $\dim(\mathcal{Y}_{i_0})>2i_0-n$, then we decode $\mathcal{Y}$ into the unique $\mathcal{F}$ such that $\mathcal{Y}_{i_0}\subseteq\mathcal{F}_{i_0}$.
\end{proposition}
\begin{proof}
Since $\dim(\mathcal{F}_{i_0}^i\cap \mathcal{F}_{i_0}^j)=2i_0-n$ for $i\neq j$ and $\dim(\mathcal{Y}_{i_0})>2i_0-n$, there exists a unique subspace $\mathcal{F}_{i_0}$ that contains $\mathcal{Y}_{i_0}$. Consequently, we are able to decode $\mathcal{Y}$ into the unique $\mathcal{F}$, such that $\mathcal{Y}_{i_0}\subseteq\mathcal{F}_{i_0}$.
\end{proof}

We consolidate all the preceding results into the following algorithm.

\begin{table}[H]
\scalebox{1}{%

	\begin{tabular}{l}
		\hline
		\textbf{Decoding algorithm}                             \\ \hline
		\textbf{Data:}  A sequence of subspaces $\mathcal{X} =(\mathcal{X}_1,\dots, \mathcal{X}_{n-1})$.            \\
		\textbf{Result:} The sent flag $\mathcal{F}\in \mathcal{C}$.                              \\
		\textbf{for} $1 \leqslant i \leqslant n-1$                                        \\
		\ \ \ \ \ \ \textbf{if} $i \leqslant k_1$ and $\dim(\mathcal{X}_i) > 0$,                    \\
		\ \ \ \ \ \ \ \ \ decode $\mathcal{X}_i$ into the only $\mathcal{F}_i\in \mathcal{C}_i$ that contains $\mathcal{X}_i$,        \\
		\ \ \ \ \ \ \ \ \ \textbf{return} the unique flag $\mathcal{F}\in \mathcal{C}$ that has $\mathcal{F}_i$ as $i$-th subspace. \\
		\ \ \ \ \ \ \textbf{else} set $\mathcal{Y}_i=\oplus^i_{j=k_1+1}\mathcal{X}_j \ \textrm{for}\ i=k_1+1,\dots,n-1.$
        \\
        \ \ \ \ \ \ \ \ \ \ \ \ \textbf{if} $k_1<i \leqslant k_1+r$ and $\dim(\mathcal{Y}_i) > i-k_1$,                    \\
		\ \ \ \ \ \ \ \ \ \ \ \ \ \ \  decode $\mathcal{Y}_i$ into the only $\mathcal{F}_i\in \mathcal{C}_i$ that contains $\mathcal{Y}_i$,        \\
	    \ \ \ \ \ \ \ \ \ \ \ \ \ \ \ \textbf{return} the unique flag $\mathcal{F}\in \mathcal{C}$ that has $\mathcal{F}_i$ as $i$-th subspace. \\
		\ \ \ \ \ \ \ \ \ \ \ \ \textbf{if} $i>k_1+r$ and $\dim(\mathcal{Y}_i)>2i-n$,   \\
		\ \ \ \ \ \ \ \ \ \ \ \ \ \ \  decode $\mathcal{Y}_i$ into the only $\mathcal{F}_i\in \mathcal{C}_i$ that contains $\mathcal{Y}_i$,        \\
		\ \ \ \ \ \ \ \ \ \ \ \ \ \ \ \textbf{return} the unique flag $\mathcal{F}\in \mathcal{C}$ that has $\mathcal{F}_i$ as $i$-th subspace. \\ \hline
	\end{tabular}%
}
\end{table}

\begin{remark}
In \cite{ANS1}, a decoding algorithm for full flag codes based on spread was introduced, which is executed in two distinct steps. The decoding algorithm presented in this paper serves as an extension of that work. The first and third steps are analogous to the original approach. However, we have incorporated an additional decoding process in the second step. This enhancement arises from the newly integrated second-layer matrix B within the \textquotedblleft sandwich" structure. In light of this feature, we have optimized the original algorithm and proposed a refined three-step decoding scheme.
\end{remark}

\section{Conclusions and future works}
In this paper, we propose the use of partial spreads to construct full flag codes. This approach resembles a \textquotedblleft sandwich" structure, which consists of one layer of a companion matrix \( M \) and two layers of partial spreads associated with \( M \). The \textquotedblleft sandwich" construction leads to the derivations of ODFC, QODFC and flag codes with various distance values. Additionally, an efficient decoding algorithm for the \textquotedblleft sandwich" construction is presented at the end.

Our study opens up multiple avenues for further research and raises several pertinent questions. Below, we outline a selection of these inquiries.
\begin{itemize}
\item [(1)] In Theorem \ref{T3.5}, we explored the relationship between the cardinality of flag code $\mathcal{C}$ and $A_q(n, 2k; k)$ for $r = 0, 1, 2$. Consequently, can we find the relationship between $\mathcal{C}$ and $A_q(n, 2k; k)$ when $r \geqslant 3$?
\item [(2)] In Section 5, we introduced a decoding algorithm specifically designed for use on erasure channels. How to investigate decoding algorithms that can be effectively employed on both erasure and insertion channels simultaneously?
\item [(3)] In Definition \ref{D4.7}, we introduced a method for constructing matrices $S[i]$ corresponding to flag codes utilizing two layers of partial spreads. In an upcoming research, we will explore how to generalize this approach to develop a novel scheme for constructing matrices with multiple layers of partial spreads.
\end{itemize}

\section*{Acknowledgement}
X. Han, X. R. Li and G. Wang are supported by the National Natural Science Foundation of China (No. 12301670), the Natural Science Foundation of Tianjin (No. 23JCQNJC00050), the Fundamental Research Funds for the Central Universities of China (No. 3122024PT24) and the Graduate Student Research and Innovation Fund of Civil Aviation University of China (No. 2024YJSKC06001).

\end{document}